\documentclass[aps,prd,amsmath,amssymb,preprintnumbers,preprint,nofootinbib,a4paper]{revtex4}
\pdfoutput=1
\usepackage{amsthm}
\usepackage{graphicx}
\usepackage{color}
\newcommand{\beq}{\begin{eqnarray}}
\newcommand{\eeq}{\end{eqnarray}}

\newcommand{\bmp}{\noindent\begin{minipage}{16cm}}
\newcommand{\emp}{\end{minipage}\vskip 7mm} 

\newcommand{\wt}{\widetilde}

\usepackage{dcolumn}
\usepackage{bm}
\usepackage{bbm}
\usepackage{subfigure}
\usepackage{pxfonts}

\theoremstyle{definition}
\newtheorem*{Def}{Definition}
\theoremstyle{plain}
\newtheorem*{thm}{Theorem}

\usepackage{epsfig}

\usepackage[ margin=5pt, font=normalsize,labelfont=bf,justification=raggedright]{caption}
\usepackage{youngtab}
\usepackage{slashed}

\definecolor{rossoCP3}{cmyk}{0,.88,.77,.40}

\usepackage{fancyhdr}
\def\lsim{\mathrel{\rlap{\lower4pt\hbox{\hskip1pt$\sim$}}
    \raise1pt\hbox{$<$}}}                
\def\gsim{\mathrel{\rlap{\lower4pt\hbox{\hskip1pt$\sim$}}
    \raise1pt\hbox{$>$}}}                

\newcommand{\drawsquare}[2]{\hbox{%
\rule{#2pt}{#1pt}\hskip-#2pt
\rule{#1pt}{#2pt}\hskip-#1pt
\rule[#1pt]{#1pt}{#2pt}}\rule[#1pt]{#2pt}{#2pt}\hskip-#2pt
\rule{#2pt}{#1pt}}

\newcommand{\Yfund}{\raisebox{-.5pt}{\drawsquare{6.5}{0.4}}}

\begin{document}
\title{\Large  \color{rossoCP3} ~~\\ Dual of QCD with One Adjoint Fermion}
\author{Matin {\sc Mojaza}$^{\color{rossoCP3}{\varheartsuit}}$}\email{mojaza@cp3.sdu.dk} 
\author{Marco {\sc Nardecchia}$^{\color{rossoCP3}{\varheartsuit}}$}\email{nardecchia@cp3.sdu.dk} 
\author{Claudio {\sc Pica}$^{\color{rossoCP3}{\varheartsuit}}$}\email{pica@cp3.sdu.dk} 
\author{Francesco {\sc Sannino}$^{\color{rossoCP3}{\varheartsuit}}$}\email{sannino@cp3.sdu.dk} 
\affiliation{
$^{\color{rossoCP3}{\varheartsuit}}${\color{rossoCP3} {\large C}entre for {\large P}article {\large  P}hysics  {\large  P}henomenology}, {\color{rossoCP3}\large CP}$^{ \bf \color{rossoCP3} 3}${ \color{rossoCP3}-Origins}, \\
University of Southern Denmark, Campusvej 55, DK-5230 Odense M, Denmark.
}

\begin{abstract}
We construct the magnetic dual of QCD with one adjoint Weyl fermion. The dual is a consistent solution of the 't Hooft anomaly  matching conditions, allows for flavor decoupling and remarkably constitutes the first nonsupersymmetric dual valid for {\it any} number of colors. The dual allows to bound the anomalous dimension of the Dirac fermion mass operator to be less than one in the conformal window.
\\[.1cm]
{\footnotesize  \it Preprint: CP$^3$-Origins-2011-01}
 \end{abstract}

\maketitle
\thispagestyle{fancy}

\section {Introduction}

One of the most fascinating possibilities is that generic asymptotically free gauge theories have magnetic duals. In fact, in the late nineties, in a series of  ground breaking papers Seiberg  \cite{Seiberg:1994bz,Seiberg:1994pq} provided strong support for the existence of a consistent picture of such a duality within a supersymmetric framework. Supersymmetry is, however, quite special and the existence of such a duality does not automatically imply the existence of nonsupersymmetric duals. One of the most relevant  results put forward by Seiberg  has been the identification of the boundary of the conformal window for supersymmetric QCD as function of the number of flavors and colors. 
The dual theories proposed by Seiberg pass a set of mathematical consistency relations known as 't Hooft anomaly matching conditions \cite{Hooft}.  Another important tool has been the knowledge of the all-orders supersymmetric beta function \cite{Novikov:1983uc,Shifman:1986zi,Jones:1983ip}. Recently we provided several analytic predictions for the conformal window of nonsupersymmetric gauge theories using different approaches \cite{Sannino:2004qp,Dietrich:2006cm,Ryttov:2007cx,Sannino:2009za,Pica:2010mt,Pica:2010xq,Ryttov:2010iz,Mojaza:2010cm, Sannino:2009aw, Ryttov:2009yw}.

We initiated in \cite{Sannino:2009qc}  the exploration of the possible existence of a QCD nonsupersymmetric gauge dual providing a consistent picture of the phase diagram as function of number of colors and flavors. Arguably the existence of a possible dual of a generic nonsupersymmetric asymptotically free gauge theory able to reproduce its infrared dynamics must match the 't Hooft anomaly conditions \cite{Hooft}. We have exhibited several solutions of these conditions for QCD in \cite{Sannino:2009qc}. An earlier exploration already appeared in \cite{Terning:1997xy}. 
In \cite{Sannino:2009me} we have also analyzed theories with fermions transforming according to higher dimensional representations. Some of these theories have been used to construct sensible extensions of the standard model of particle interactions of technicolor type passing precision data and known as Minimal Walking Technicolor  models \cite{Sannino:2004qp,Dietrich:2006cm}. Other interesting studies of technicolor dynamics making use of higher dimensional representations appeared in \cite{Christensen:2005cb}. These are the  known extension of technicolor type possessing the smallest intrinsic  $S$ parameter \cite{Peskin:1990zt,Peskin:1991sw,Kennedy:1990ib,Altarelli:1990zd} while being able to display simultaneously (near) conformal behavior before adding the backreaction due to traditional type extended technicolor interactions (ETC) \cite{Eichten:1979ah}. For recent analysis of the relevant properties of the $S$ parameter, using also  duality arguments we refer to \cite{Sannino:2010ca,Sannino:2010fh,DiChiara:2010xb}. As for the issue of the effects of the ETC interactions on the technicolor dynamics and associated conformal window, it has recently been argued in \cite{Fukano:2010yv} that before adding the ETC interactions the technicolor theory, in isolation,   should already be conformal for the combined model to be phenomenologically viable. A deeper understanding of the gauge dynamics of (near) conformal gauge theories is therefore needed  making the study of gauge duals very relevant. For example, the {\it magnetic} dual allows to predict, in principle, the critical number of flavors below which the electric theory looses large distance conformality.  It has also been show in  \cite{Sannino:2010fh} that it is possible to use  gauge dualities to compute  the $S$ parameter in the nonperturbative regime of the electric theory. 
 
Here we put the idea of nonsupersymmetric gauge duality on a much firmer ground by showing that for certain scalarless gauge theories with a spectrum similar to the one of QCD the gauge dual passes a large number of consistency checks.

The theory we choose is QCD with $N_f$ Dirac flavors and one adjoint Weyl fermion. A relevant feature of this theory is that it possesses the same global symmetry of super QCD despite the fact that squarks are absent. This means that there are four extra anomaly constraints not present in the case of ordinary QCD, moreover we will show that the dual can be constructed for any number of colors greater than two. 

The magnetic dual is a new gauge theory featuring magnetic quarks and a Weyl adjoint fermion, new gauge singlet fermions which can be identified as states composite of the electric variables, as well as scalar states needed to mediate the interactions between the magnetic quarks and the gauge singlet fermions. The new scalars allow for a consistent flavor decoupling which was an important consistency check in the case of supersymmetry. 

We will show that the dual allows to bound the anomalous dimension of the Dirac fermion mass operator to be less than one in the conformal window, and also estimate the critical number of flavors below which large distance conformality is lost in the electric variables. 

\section{The Electric Theory:  QCD with one Adjoint  Weyl Fermion} 
The electric theory is constituted by a scalarless  $SU(N)$ gauge theory with $N_f$ Dirac fermions and $N$ larger than two, as in QCD, but with an extra Weyl fermion transforming according to the adjoint representation of the gauge group.
The quantum global symmetry is:  
\begin{equation}
SU_L(N_f) \times SU_R(N_f) \times U_V(1)  \times U_{AF}(1)\ . 
\end{equation}
At the classical level there is one more $U_A(1)$ symmetry destroyed by quantum
corrections due to the Adler-Bell-Jackiw anomaly. Therefore of the three independent $U(1)$ symmetries only two survive, a vector like $U_V(1)$ and an axial-like anomaly free (AF) one indicated with $U_{AF}(1)$.  The spectrum of the theory and the global transformations are summarized in table \ref{Electric}.
\begin{table}[h]
\[ \begin{array}{|c| c | c c c c | } \hline
{\rm Fields} &  \left[ SU(N) \right] & SU_L(N_f) &SU_R(N_f) & U_V(1)&U_{AF}(1) \\ \hline \hline
\lambda &{\rm Adj} & 1 &1 &~~0& ~~1 \\
Q &\Yfund &{\Yfund }&1&~~1 & -\frac{N}{N_f} \\
\widetilde{Q} & \overline{\Yfund}&1 &  \overline{\Yfund}& -1 & -\frac{N}{N_f}   \\
G_{\mu}&{\rm Adj}   &1&1  &~~0  & ~~0\\
 \hline \end{array} 
\]
\caption{Field content of the electric theory and field tranformation properties. The squared brackets refer to the gauge group.}
\label{Electric}
\end{table}

The  global anomalies are: 

\begin{eqnarray}
SU_{L/R}(N_f)^3 = \pm N \ , \quad 
U_{V}(1) \, SU_{L/R}(N_f)^2  =  \pm \frac{N}{2} \ , \quad
U_{AF}(1) \, SU_{L/R}(N_f)^2  = - \frac{N^2}{2N_f}  \ ,\label{eq:An1}\\  
U_{AF}(1)^3   =    N^2 - 1 - 2\frac{N^4}{N_f^2} \ , \quad  
U_{AF}(1)\, U_{V}(1)^2 = - 2N^2 \ , \quad 
 [{\rm Gravity}]^2U_{AF}(1) =  -(N^2 + 1) \label{eq:An2} 
\end{eqnarray}
The first two anomalies are the same as in QCD and they are associated to the triangle diagrams featuring three $SU(N_f)$ generators (either all right or all left) at the vertices , or two 
$SU(N_f)$ generators (all right or all left) and one $U_V(1)$ charge.

\section{Solutions of the 't Hooft anomaly conditions for any number of colors} 

 We seek solutions of the anomaly matching conditions in the conformal window of the electric theory. This means that we consider a sufficiently large number of Dirac flavors so that the electric coupling constant freezes at large distances.

{}Following Seiberg we assume the dual to be a new $SU(X)$ gauge theory with global symmetry group $SU_L(N_f)\times SU_R(N_f) \times U_V(1) \times U_{AF}(1)$  featuring 
{\it magnetic} quarks ${q}$ and $\widetilde{q}$ transforming in the fundamental representation of $SU(X)$ together with a magnetic Weyl fermion $\lambda_{m}$ in the adjoint representation. We also introduce a minimal set of gauge singlet fermionic particles which can be viewed as composite states in terms of the electric variables but are considered elementary in the dual description. We limit to gauge singlet fermionic states which can be interpreted as made by three electric fields for any number of colors \footnote{Note that for $N=3$ ordinary baryons are also made by three states, however since we will find consistent solutions of the 't Hooft anomaly matching conditions for any $N$ we have not included them in the spectrum.}. The idea behind this choice is that composite states made by more fields could be constructed also in the dual theory from the new elementary fields. 

We summarize  the dual spectrum in table \ref{QCDAdual}. The dual gauge group can be different from the electric one since the only physical quantities are gauge singlets.  On the other hand the global symmetries must be the original 
ones\footnote{There might be exceptions such as possible enhanced global symmetries but it is not the case here.} given that can be physically probed.  

We wish to find solutions to the 't Hooft anomaly conditions valid for any number of colors. 
\begin{table}[b!]
\[ \begin{array}{|c|c|c|c c c c|c|} \hline
{\rm Fields} & \text{Composite eq.} &\left[ SU(X) \right] & SU_L(N_f) &SU_R(N_f) & U_V(1)& U_{AF}(1) & \# ~{\rm  of~copies} \\ \hline 
\hline 
\lambda_m & - & {\rm Adj} & 1 & 1 & 0 & z' & 1 \\
 q & - &\Yfund &\overline{\Yfund }&1&~~y & - \frac{X}{N_f} z &1 \\
\widetilde{q} & - & \overline{\Yfund}&1 &  {\Yfund}& -y& - \frac{X}{N_f} z  &1   \\
 M & Q \lambda \widetilde{Q} & 1 & \Yfund & \overline{\Yfund} & 0 & 1-2\frac{N}{N_f} & \ell_{M}  \\
 \widetilde{M} & \overline{Q \widetilde{Q}}\lambda & 1 &  \overline{\Yfund} &\Yfund & 0 & 1+2\frac{N}{N_f} &  \ell_{\widetilde{M}} \\
 \Lambda_S  & \overline{Q \lambda} Q \, ,    \overline{ \widetilde{Q} \lambda} \widetilde{Q} \, , \overline{\lambda \lambda } \lambda    & 1 & 1 & 1 & 0 & -1 &  \ell_{\Lambda_S} \\
  \Lambda_L & \overline{Q \lambda} Q & 1 & {\rm Adj}  & 1 & 0 & -1 & \ell_{\Lambda_L} \\
  \Lambda_R & \overline{\widetilde{Q} \lambda} \widetilde{Q} & 1 & 1 & {\rm Adj} & 0 & -1 & \ell_{\Lambda_R}  \\
  \Lambda_G    & \lambda GG& 1 & 1 & 1 & 0 & 1 & \ell_{\Lambda_G}\\
 \Lambda & \lambda \lambda \lambda & 1 & 1 & 1 & 0 & 3 & \ell_{\Lambda}\\
 \hline \end{array} 
\]
\caption{Massless spectrum of the candidate {\it magnetic}  fermions and their  transformation properties under the dual gauge and global symmetry group. The last column represents the multiplicity of each state and each state is a  Weyl fermion in the $(1/2,0)$ representation of the Lorentz group.}\label{QCDAdual}
\end{table}
$z^{\prime}$, $z$ and $y$ are arbitrary Abelian charges for the magnetic quarks and adjont fermion while the ones of the $SU(X)$ magnetic singlets are derived from the electric constituents. From table \ref{QCDAdual} we write below all the relevant anomalies for the dual theory which we require to match the electric ones given in eqs. \eqref{eq:An1} and  \eqref{eq:An2}  
\begin{align}
[SU(X)]^2 U_{AF}(1) = &X z'  + \frac{1}{2}\left(- \frac{X}{N_f} z\right) N_f \times 2 = X (z'-z) 
:= 0\\
SU_{L/R} (N_f)^3  = &\pm [  N_f ( \ell_{M} -  \ell_{\widetilde{M}} ) -X] 
:= \pm N\\
SU_{L/R} (N_f)^2 U_V(1) =& \pm \frac{1}{2} y X 
:= \pm \frac{N}{2} \\
SU_{L/R} (N_f)^2 U_{AF}(1) = &\frac{1}{2}  \left(- \frac{X}{N_f} z\right) X + \frac{1}{2} \ell_M N_f (1-2 \frac{N}{N_f}) +  \frac{1}{2} \ell_{\widetilde{M}} N_f (1+2 \frac{N}{N_f}) - \ell_{\Lambda_{L/R}} N_f   \nonumber \\
&:= - \frac{N^2}{2 N_f} 
\end{align}
 \begin{align}
U_V(1)^2 U_{AF}(1) = &2 \times y^2 \left(- \frac{X}{N_f} z\right) X N_f 
:= - 2 N^2\\
U_{AF}(1)^3 = &(X^2-1){z'}^{3} + X N_f\left(- \frac{X}{N_f} z\right)^3 \times 2 + \ell_M N_f^2 (1-2 \frac{N}{N_f})^3 + \ell_{\widetilde{M}} N_f^2 (1 + 2\frac{N}{N_f})^3 \nonumber \\
 &+ \ell_{\Lambda_L} (N_f^2-1) (-1)^3 +  \ell_{\Lambda_R} (N_f^2-1) (-1)^3 + \ell_{\Lambda_G} + (-1)^3 \ell_{\Lambda_S} + 3^3 \ell_\Lambda \nonumber \\
& := N^2 - 1 - 2 \frac{N^4}{N_f^2}\\
 [{\rm Gravity}]^2 U_{AF}(1) = &(X^2-1)z' + X N_f\left(- \frac{X}{N_f} z\right) \times 2 + \ell_M N_f^2 (1-2 \frac{N}{N_f}) + \ell_{\widetilde{M}} N_f^2 (1 + 2\frac{N}{N_f}) \nonumber \\
 &- \ell_{\Lambda_L} (N_f^2-1) -  \ell_{\Lambda_R} (N_f^2-1)+ \ell_{\Lambda_G} - \ell_{\Lambda_S} + 3 \ell_\Lambda 
 := -N^2 - 1 \ .
\end{align}

The electric $U_{AF}(1)$, anomaly free condition, i.e.  $[SU(X)]^2 U_{AF}(1) = 0$ provides the first nontrivial constraint, i.e. $z'=z$.

Consider now the 't Hooft anomaly matching conditions from $SU_{L}(N_f)^2 U_V(1)$ and the $U_V(1)^2 U_{AF}(1)$:
\begin{equation}
\frac{Xy}{2}=\frac{N}{2} \ , \qquad  
X^2y^2z=N^2 \ .
\end{equation}
which can be simultaneously solved for 
\begin{equation}
z=1 \qquad {\rm and} \qquad  y=N/X \ .
\end{equation} 
The difference between the $SU_{L}(N_f)^2 U_{AF}(1)$ and $SU_{R}(N_f)^2 U_{AF}(1)$ anomalies forces \mbox{$\ell_{\Lambda_R}=\ell_{\Lambda_L}$}.

At this point we remain with 7 unknowns $\left(X,\ell_{\Lambda_M},\ell_{\Lambda_{\widetilde{M}}},\ell_{\Lambda_L},\ell_{\Lambda_S},\ell_{\Lambda_G},\ell_{\Lambda} \right)$  and 4 independent anomaly matching conditions. We can use the 4 equations to solve for $\ell_{\Lambda_M},\ell_{\Lambda_{\widetilde{M}}},\ell_{\Lambda_L}$ and $\ell_{\Lambda}$. Notice that the system is linear in these 4 unknowns. Obviously we expect that the solution is not unique and we can parametrize the set of all the solutions using combinations of the remaining three unknowns: $X,\ell_{\Lambda_S},\ell_{\Lambda_G}$. Before showing the result we already know that the solution of the system can depend only on $X$ and linearly on $z_D \equiv \ell_{\Lambda_G}-\ell_{\Lambda_S}$. The solution depends only on $z_D$  because $\Lambda_S$ and $\Lambda_{G}$ form vector like pairs. The sum $\ell_{\Lambda_G}+\ell_{\Lambda_S}$ is not constrained by the 't Hooft anomaly matching conditions.

Here is the solution:
\begin{align}
\ell_\Lambda &= \frac{ \left( N+X \right)\left( 3N-X \right) \left( N+ N_f+X \right)  \left( X +N- N_f \right) -6 z_D  N_f^{2} {N}^{2}} {6 {N_f}^{2} \left( 3{N}^{2}-2 \right) } \label{eq:lL}\\
\ell_{\Lambda_L} &= - \frac{ \left( N+X \right)\left( 3N-X \right) \left( N+ N_f+X \right)  \left( X +N- N_f \right) -4 z_D  N_f^{2}} {4 {N_f}^{2} \left( 3{N}^{2}-2 \right) } \label{eq:lLL}\\
\ell_{M} &=  \frac{ \left( N+X \right)\left( 3 N^2 + 3 N N_f - 2 N X - 4 + X^2 - N_f X \right)  \left( X +N + N_f \right) +4 z_D  N_f^{2}} {4 {N_f}^{2} \left( 3{N}^{2}-2 \right) }\label{eq:lLM}\\
\ell_{\widetilde{M}} &=  \frac{ \left( N+X \right)\left( 3 N^2 - 3 N N_f - 2 N X - 4 + X^2 + N_f X \right)  \left( X +N - N_f \right) +4 z_D  N_f^{2}} {4 {N_f}^{2} \left( 3{N}^{2}-2 \right) }
\end{align}
The four equations are further constraint by being nonnegative integers. 
The problem is then reduced to finding combinations of $X$ and $z_D$ fulfilling this requirement.  

In this way we have completely characterized \textit{all} the possible solutions of the 't Hooft anomaly matching in a dual theory with matter content presented in Table \ref{QCDAdual}. We notice that one can reproduce a Seiberg-like solution for $\ell_M =1$ and the remaining indices set to zero. 

As we shall show, the above equations can be rewritten in a more transparent way. The following linear combination of the first two equations \eqref{eq:lL} and \eqref{eq:lLL}  gives:
\begin{eqnarray}
z_D = -3 \ell_\Lambda -2\ell_{\Lambda_L} \ .\label{eqzd}
\end{eqnarray}
Therefore, imposing positivity of the multiplicities, we can only have solutions for which $z_D \leq 0$.  Furthermore, we can also form another linear combination independent on $z_D$:
\begin{eqnarray}
\ell_\Lambda + N^2\ell_{\Lambda_L} = \frac{ \left( N+X \right)\left( X-3N \right) \left(X+N+ N_f \right)  \left( X +N- N_f \right)} {12 {N_f}^{2} } \ . \label{twolambda}
\end{eqnarray}
It is more convenient to work with the sum and difference of $\ell_{M}$ and $\ell_{\widetilde{M}}$:
\begin{align}
\ell_M-\ell_{\widetilde{M}}&=\frac{N+X}{N_f}\label{eq:3N} = d\\
\ell_M+\ell_{\widetilde{M}}&= d^2 + 2 \ell_{\Lambda_L} ,
\end{align}
where we have defined $d$ in the first equation. $d$ must both be a positive integer implying $X \geq N_f - N$.  The positivity condition of eq. \eqref{twolambda}  further requires the inequality $X\geq 3N$ when   $X >N_f - N$.

In summary, the original set of  't Hooft anomaly matching conditions can be neatly rewritten as:
\begin{eqnarray}
\ell_M-\ell_{\widetilde{M}}&=& d \label{eq:3N} \\
\ell_M+\ell_{\widetilde{M}}&=& d^2 + 2 \ell_{\Lambda_L} \\
\ell_\Lambda &=&  \frac{(d-1) \ d \ (d+1) (d \ N_f - 4 N ) N_f}{12} -  N^2 \ell_{\Lambda_L} \label{eq:f}\\
z_D &=& -3 \ell_\Lambda -2\ell_{\Lambda_L} \ .
\end{eqnarray}
In this form, the constraint on the $\ell$'s being positive integers is always fulfilled 
and thus all solutions are given in terms of the \emph{generic parameters} $d$ and $\ell_{\Lambda_L}$. One immediate consequence is that the dual gauge group $SU(X)$ is $SU(d \ N_f - N)$ with $d$  a positive integer. As we shall show in the Appendix \ref{involutionApp} this structure of the gauge group admits an involution of the duality operation, meaning that if one dualizes again it is possible to recover the $SU(N)$ group. This, however, is only a necessary condition for the complete involution of the duality operation to hold exactly, meaning that one has still to demonstrate that, at the fixed point the dual of the dual is actually the electric theory.

The first, and most relevant solution, is obtained for $d=1$ and it mimics Seiberg's solution for the dual of super QCD. We have, in fact, that:
  \begin{equation}
d  = 1 \quad \Rightarrow \quad X = N_f - N \ , \quad \ell_{M} = 1 \ , \quad 
\ell_\Lambda = \ell_{\Lambda_L} = \ell_{\widetilde{M}} = z_D = 0 \ . 
\end{equation}
  No other solutions exist for $d=1$, i.e. $\ell_{\Lambda_L}$ vanishes. The dual gauge group is therefore $SU(N_f - N)$.

%

We now show that the minimal value of $-z_D = \ell_{\Lambda_S} - \ell_{\Lambda_G}= 2 \ell_{\Lambda_L} + 3 \ell_{\Lambda}$ different from zero is $5$. Solutions for which $-z_D = 2 \ , 3 \ , 4$ require either $\ell_\Lambda$ or $\ell_{\Lambda_L}$
to be zero, but one finds from \eqref{eq:f}, that:
\begin{align}
\text{for} &\quad \ell_\Lambda = 0 \quad \Rightarrow \quad min(\ell_{\Lambda_L}) = 3 \\
\text{for} &\quad \ell_{\Lambda_L} = 0 \quad \Rightarrow \quad min(\ell_{\Lambda}) = 5\\
\text{thus}  &\quad min(-z_D) = 5 \quad \text{with} \quad \ell_\Lambda = \ell_{\Lambda_L}=1.
\end{align}
For any $d$ greater than one for which $\ell_{\Lambda_L}$ and $\ell_\Lambda$ assume values lower than the ones provided above there is no solution to the 't Hooft anomaly matching conditions. This can be shown by first noting that by setting one of the two $\ell$s to zero in \eqref{eq:f} the minimal value of the other $\ell$ is always obtained by setting $d=2$. The second step is to set $\ell_\Lambda$ to zero and therefore obtain: 
\begin{equation}
\left( \frac{N_f}{N} - 2\right) \frac{N_f}{N} = \ell_{\Lambda_L} \ . 
 \end{equation}
 It is clear that the minimum $\ell_{\Lambda_L}$ must be three for this equation to hold. When setting instead $\ell_{\Lambda_L}=0$ we obtain: 
 \begin{equation}
 \ell_\Lambda = \left(N_f - 2 N \right)N_f  \ .
 \end{equation}
 This equation, given that the $\ell$s are positive  integers, requires $N_f > 2N$ and therefore the minimum $N_f$ must be:
 \begin{equation}
 N_f = 2 N+ 1 \ , 
 \end{equation}
 which means: 
 \begin{equation}
 \ell_\Lambda= 2N +1 \ ,
 \end{equation}
  and therefore $\ell_\Lambda > 5$.  
If both $\ell_\Lambda$ and $\ell_{\Lambda_L}$ are equal to zero one \emph{either} recovers
the $d=1$ or one can saturate the inequality $X\geq 3N$ and
find the following  solutions:
\begin{equation}
\text{for} \quad d >1 \quad \text{and} \quad \ell_\Lambda = \ell_{\Lambda_L} = z_D =0 \quad \Rightarrow \quad 
X = 3 N \ ,\quad 
d  = 4 \frac{N}{N_f} \ , \quad 
\ell_M = \frac{d^2+d}{2} \ , \quad 
\ell_{\widetilde{M}} = \frac{d^2 - d}{2} 
\end{equation}
These solutions are, however, unnatural since $X$ is $N_f$ independent and furthermore $d$ is by definition restricted to be an integer. A solution of this type implies that in the conformal window it is impossible to change the number of flavors in the electric theory and remain with an integer $d$ for the same number of colors assuming one had started with a number of flavors for which $d$ was an integer. Therefore, we do not consider these solutions to represent viable magnetic duals.  Solutions with $z_D \leq -5$ are certainly non-minimal strongly indicating that the $d=1$ solution is the relevant one.

Requiring the electric theory to be asymptotically free imposes further general constraints. We start with recalling the first coefficient of the beta function for the
electric theory which is:
\begin{align}
\beta_0^e = 3 N - \frac{2}{3} N_f  \ .
\end{align}
For the theory to be asymptotically free, one imposes $\beta_0^e \geq 0$,
leading to the constraint on $N_f$:
\begin{align}
N_f \leq \frac{9}{2} N \ .
\end{align}
{}For $\ell_\Lambda \geq 0$  from eq. \eqref{eq:f} we find the condition:
\[
\frac{(d-1) \ d \ (d+1) (d \ N_f - 4 N ) N_f}{12} \geq  N^2 \ell_{\Lambda_L} \ ,
\]
leading to the constraint, already noted earlier: 
\[
\text{for } d\neq 1 \ , \quad d\geq 4 \frac{N}{N_f} \ .
\]

We can set another constraint on $\alpha = {N_f}/{N}$ rewriting the inequality above as:
\[
\alpha^2 - \frac{4}{d}\alpha - \frac{12 \ell_{\Lambda_L}}{(d^2-1)d^2} \geq 0,
\]
which leads to
\[
\alpha = \frac{N_f}{N} \geq 2 \frac{d^2 -1 + \sqrt{(d^2-1)(d^2 -1 + 3 \ell_{\Lambda_L})}}{(d^2-1)d} \ .
\]
Combining the two constraints on $N_f/N$ we obtain: 
\begin{equation}
2 \frac{d^2 -1 + \sqrt{(d^2-1)(d^2 -1 + 3 \ell_{\Lambda_L})}}{(d^2-1)d} \leq \frac{N_f}{N} \leq \frac{9}{2}\quad
\text{with} \quad d \geq 4 \frac{N}{N_f}  \ .
\end{equation}

Expressing instead the bound on $\ell_{\Lambda_L}$, we see that it is bounded from above,
hence constraining the number of possible solutions to be finite for a fixed value of $d$:
\[
\ell_{\Lambda_L} \leq \frac{1}{12} d \ (d^2-1)(d \frac{N_f}{N}-4) \frac{N_f}{N}.
\]

The most natural solution is the one with the lowest value assumed by all indices which corresponds to $d=1$ and therefore has associated gauge group $SU(N_f - N)$.  We will, therefore, concentrate on this theory below and summarize here the fermionic spectrum which resembles the supersymmetric version of the theory. 

\begin{table}[h]
\[ \begin{array}{|c|c|c c c c|} \hline
{\rm Fields} &\left[ SU(N_f-N) \right] & SU_L(N_f) &SU_R(N_f) & U_V(1)& U_{AF}(1) \\ \hline 
\hline 
\lambda_m & {\rm Adj} & 1 & 1 & 0 & 1 \\
 q &\Yfund &\overline{\Yfund }&1&~~\frac{N}{N_f - N} & - \frac{N_f - N }{N_f}  \\
\widetilde{q}& \overline{\Yfund}&1 &  {\Yfund}& -\frac{N}{N_f - N}& - \frac{N_f - N}{N_f}     \\
 M  & 1 & \Yfund & \overline{\Yfund} & 0 & 1-2\frac{N}{N_f} \\
 \hline \end{array} 
\]
\caption{Massless spectrum of  the {\it magnetic} fermions and their  transformation properties under the $SU(d N_f - N)$ dual gauge  (with $d=1$) and the global symmetry group.}\label{QCDAdualm}
\end{table}

\section{Decoupling of flavors}
Following Seiberg a consistent dual description requires that decoupling of a flavor in the electric theory corresponds to decoupling of a flavor in the magnetic theory. The resulting magnetic theory is still the dual of the electric theory with one less flavor. 
For nonsupersymmetric theories one can still expect a similar phenomenon to occur, if no phase transition takes place as we increase the mass of the specific flavor we wish to decouple. In fact, for nonsupersymmetric theories, this idea is similar to require the {\it mass persistent} condition used by Preskill and Weinberg \cite{Preskill:1981sr} according to which one can still use the 't Hooft anomaly matching conditions for theories with one extra massless flavor to constrain the solutions with one less flavor if, when giving mass to the extra flavor, no phase transition occurs. We now provide a time-honored example of the use of the {\it mass persistent} condition.  For example if one starts with three flavors QCD the 't Hooft anomaly conditions cannot be satisfied by massless baryonic states and therefore chiral symmetry must break. However for two flavors one can find a solution and therefore one cannot decide if chiral symmetry breaks unless the two-flavor case is embedded in the three flavor case. Using the three flavors anomaly conditions to infer that chiral symmetry must break for the two flavors case requires that as we take the {\it strange} quark mass large compared to the intrinsic scale of the theory no phase transition occurs apart from the explicit breaking of the flavor symmetry.  A counter example is QCD at nonzero baryonic matter density for which such a phase transition is expected to occur and therefore two-flavors QCD  can be realized without the breaking of the flavor symmetry \cite{Sannino:2000kg}.  A more detailed analysis of the validity of the 't Hooft anomaly conditions at nonzero matter density appeared in \cite{Hsu:2000by}.

To investigate the decoupling of each flavor one needs to introduce bosonic degrees of freedom. These are not constrained by anomaly matching conditions but are kept massless by the requirement that the magnetic and electric theories must display large distance conformality. Interactions among the mesonic degrees of freedom and the fermions in the dual theory cannot be neglected in the regime when the dynamics is strong.   For this, 
we need to add Yukawa terms in the dual Lagrangian. Here we investigate the case $d=1$ when decoupling 
a flavor in the electric theory.  The diagram below shows how the non-abelian global and gauge symmetries of the electric and magnetic theories   change upon decoupling of a Dirac flavor which is indicated by a down arrow for both theories:
\begin{align}
\begin{array}{ccc}
{Electric} & dualizing & Magnetic \\[2mm]
[SU(N)] \times SU_L(N_f) \times SU_R (N_f) &\longrightarrow & [SU(N_f - N)] \times SU_L(N_f) \times SU_R(N_f)  \\[2mm]
\Downarrow & & \Downarrow \\[2mm]
[SU(N)]  \times SU_L(N_f-1)\times  SU_R(N_f-1)  & \longrightarrow & [SU(N_f-1-N)]  \times SU_L(N_f-1)\times  SU_R(N_f-1)  \nonumber 
\end{array}
\end{align}
The abelian symmetries $U_V(1) \times U_{AF}(1)$ remain intact. It is clear from the diagram above that to ensures duality the decoupling of a flavor in the magnetic theory must also entail a breaking of the dual gauge symmetry.
The dual  theory  is vector-like and therefore the Vafa-Witten theorem \cite{Vafa:1983tf} forbids the spontaneous breaking of the magnetic gauge group. We
are, then, forced to introduce colored scalar fields that break the symmetry through a Higgs-mechanism.

\emph{How do we introduce the correct scalar spectrum and associated Yukawa terms for the magnetic dual?}

We start by identifying the part of the magnetic spectrum which must acquire a mass term when adding an explicit mass term in the electric theory for the $N_f$-th flavor. 
We therefore consider the following group decomposition:
\[\small
\begin{array}{| c | c | c l |} \hline
 & \{ SU(X) \mid SU_L(N_f) , SU_R(N_f) \} & \{ SU(X-1) \mid SU_L(N_f-1) , SU_R(N_f-1) \} &  \\ \hline
 \hline
\lambda_m & \{ \text{Adj} \mid 1 , 1\} & \{ \text{Adj} \mid  1, 1\}\oplus \{ \Yfund  \mid  1, 1\} \oplus \{ \overline{\Yfund} \mid  1, 1\} \oplus \{  1 \mid  1, 1\} & = \lambda_m \oplus \lambda'_{m,1} \oplus \lambda'_{m,2} \oplus \lambda_m^S\\[2mm]
 q & \{\Yfund \mid \overline{\Yfund}, 1 \} & \{\Yfund \mid \overline{\Yfund}, 1 \} \oplus \{\Yfund \mid 1, 1 \} \oplus \{1 \mid \overline{\Yfund}, 1 \} \oplus \{1 \mid 1, 1 \} & = q \oplus q'_1 \oplus q'_2 \oplus q^S \\[2mm]
  \widetilde{q} & \{\overline{\Yfund } \mid 1, \Yfund \} & \{\overline{\Yfund} \mid 1, \Yfund \} \oplus \{\overline{\Yfund} \mid 1, 1 \} \oplus \{1 \mid 1, \Yfund \} \oplus \{1 \mid 1, 1 \} & = \widetilde{q} \oplus \widetilde{q}'_1 \oplus \widetilde{q}'_2 \oplus \widetilde{q}^S\\[2mm]
 M & \{ 1 \mid \Yfund, \overline{\Yfund} \} & \{ 1 \mid \Yfund, \overline{\Yfund} \} \oplus \{1 \mid \Yfund,  1 \} \oplus \{1 \mid 1, \overline{\Yfund} \} \oplus \{1 \mid 1, 1 \} & = M \oplus M'_1 \oplus M'_2 \oplus M^S\\[1mm]
 \hline
\end{array}
\]
which clearly indicates that the flavor decoupling in the electric theory must lead to the generation of mass terms for the following states 
\begin{equation}
\lambda'_{m,1}, ~~\lambda'_{m,2},~~ q'_1, ~~q'_2,~~ \widetilde{q}'_1, ~~\widetilde{q}'_2, ~~M'_1,~~ M'_2,~~
\lambda_m^S, ~~q^S,~~ \widetilde{q}^S, ~~M^S\ ,
\end{equation} 
for the duality to be consistent.
We will introduce mass terms through a Higg-mechanism  via scalar fields acquiring the appropriate vacuum expectation values (vev) in order to induce the correct mass terms. To identify the scalar spectrum  we consider Yukawa interactions such as
 $\phi M q$, with $\phi$ a new scalar field which should generate a mass term, upon $\phi$ condensation, of the type $q'_2  M'_1$ and  $q^S M^S$. Note that the fermionic singlet states also receive mass terms via the same couplings.
To provide a complete decoupling of all the non-singlet states the following unique set of  quantum numbers for the complex scalars is needed: 
\begin{align}
\lambda_m \widetilde{q} &= \{\overline{\Yfund} \mid 1, \Yfund, -y, \frac{N}{N_f} \} \quad \longrightarrow \phi = \{\Yfund \mid 1, \overline{\Yfund}, y, -\frac{N}{N_f} \} \\
\lambda_m q &=  \{\Yfund \mid \overline{\Yfund}, 1, y, \frac{N}{N_f} \} \quad ~~~\longrightarrow \widetilde{\phi} = \{\overline{\Yfund} \mid \Yfund, 1, -y , -\frac{N}{N_f} \} \\
M q & = \{\Yfund \mid 1, \overline{\Yfund}, y, -\frac{N}{N_f} \} \quad \longrightarrow \phi_3 = \{ \overline{\Yfund} \mid 1, \Yfund, - y, \frac{N}{N_f} \} \sim \phi^*\\
M \widetilde{q} & = \{\overline{\Yfund} \mid \Yfund, 1, -y, -\frac{N}{N_f} \} \quad \longrightarrow \phi_4 = \{ \Yfund \mid \overline{ \Yfund}, 1 , y, \frac{N}{N_f} \} \sim \widetilde{\phi}^*
\end{align}
Where on the left we showed the fermionic bilinear we wish to generate and on the right the needed corresponding complex scalar. 

The similarity of the two last scalar fields with the conjugate of the first two fields is
at the level of quantum numbers only. As a minimalistic approach, however, we choose them to be pairwise equivalent, and thus reduce the number of degrees of freedom. The Yukawa terms giving mass to the non-singlet states are thus:
\[
\mathcal{L}_Y = \widetilde{y}_\lambda \phi \lambda_m \widetilde{q} + y_\lambda \widetilde{\phi} \lambda_m q + y_M \phi^* M q + \widetilde{y}_M \widetilde{\phi}^* M \widetilde{q} + {\rm h.c.} \ .
\] 

These terms automatically give masses to the singlet states as well.  However, each singlet receives contribution from two terms. We now 
must ensure that there will be no combination of singlet states with
vanishing mass.

Denote the vev's of the $\phi$'s with:
\begin{equation}
\langle \phi \rangle = \phi_{X}^{\overline{N_f}} = v\\ \ , \qquad 
\langle \widetilde{\phi} \rangle = \widetilde{\phi}^{X}_{N_f} = \widetilde{v} \ ,
\end{equation}
with $X$  and $N_f$ the  particular dual color and flavor indices along which we align the scalar condensates. 

When the scalar fields assume these vev's, the Yukawas reduce to:
\begin{align}
\mathcal{L}_Y =& \wt{y}_\lambda v \left( \lambda'_{m,1}\wt{q}'_1 + \lambda_m^S \wt{q}^S \right)
+y_\lambda \wt{v}  \left( \lambda'_{m,2}q_1 + \lambda_m^S q^S \right) +
y_M v^* \left( M'_1q_2 + M^S q^S \right) + \wt{y}_M \wt{v}^* \left( M'_2 \wt{q}'_2 + M^S \wt{q}^S \right) + {\rm h.c.} \ .
\end{align}
Thus, the singlet mass-matrix is
\begin{align}
\mathcal{M}_S = 
\begin{array}{c} \lambda_m^S \\[1mm] q^S \\[1mm] \wt{q}^S \\[1mm] M^S \end{array}
\begin{pmatrix}
0 &  y_\lambda \wt{v} & \wt{y}_\lambda v & 0  \\[1mm]
 y_\lambda \wt{v}  & 0 & 0 & y_M v^* \\[1mm]
\wt{y}_\lambda v & 0 & 0 & \wt{y}_M \wt{v}^* \\[1mm]
0 &   y_M v^*  &  \wt{y}_M \wt{v}^* & 0 
\end{pmatrix},
\end{align}
with
\[
\det \mathcal{M}_S = \left(y_\lambda \wt{y}_M \mid \wt{v} \mid^2 - \wt{y}_\lambda y_M \mid v \mid^2 \right )^2.
\]
All symmetries are correctly broken if and only if
\[
y_\lambda \wt{y}_M \mid \wt{v} \mid^2 \neq \wt{y}_\lambda y_M \mid v \mid^2 \ ,
\]
and furthermore all the unwanted singlets correctly decouple.

We are still missing an essential ingredient, i.e. the possibility to communicate to the magnetic theory   the introduction of a mass term in the electric theory for some of the flavors.   The scalar vevs are expected to be induced by the electric quark masses. The solution to this problem  is suggested by inspecting the electric mass term which reads: 
\begin{equation}
{\rm Tr}\left[m Q\widetilde{Q}\right] + {\rm h.c.} \equiv {\rm Tr}\left[m {\Phi_{Q\widetilde{Q} }} \right] + {\rm h.c.} \ ,
\end{equation}
with $m$ the explicit mass matrix. If we wish to decouple a single flavor then the mass matrix  has the only entry $m \delta_{N_f}^{\bar{N}^f}$ with zero for the remaining ones. We have also introduced the complex scalar $\Phi_{Q\widetilde{Q} }$  with the following quantum numbers with respect to the electric gauge group: 
\begin{equation}
\Phi_{Q\widetilde{Q} }= \{1 \mid \Yfund, \overline{\Yfund}, 0, -2\frac{N}{N_f} \} \sim  Q\widetilde{Q}   \ .
\end{equation} 
This state resembles the supersymmetric auxiliary complex scalar field of $M$ and should be identified with the standard QCD meson. Adding $\Phi_{Q\widetilde{Q} }$ as part of the dual spectrum allows us to introduce the following interaction in the dual theory: 
\begin{equation}
{\cal L}_{\Phi_{Q\widetilde{Q} }} =  \frac{y_{\Phi_{Q\widetilde{Q} }}}{\mu}\, \widetilde{\phi}^* \Phi_{Q\widetilde{Q} }{\phi}^*  + {\rm h.c.} \ . 
\label{inducedmass}
\end{equation}
This operator is the nonsupersmmetric Lagrangian equivalent of Seiberg's superpotential term. The ultraviolet scale $\mu$ serves to adjust the physical dimensions since $\phi$ and $\widetilde{\phi}$ are canonically normalized elementary magnetic fields while $\Phi_{Q\widetilde{Q} }$ has dimension three being an electric  field made by $Q\widetilde{Q}$. It was shown in \cite{Sannino:2008pz,Sannino:2008nv} that due to the presence of a nonzero mass term $\Phi_{Q\widetilde{Q} }$ acquires  a vev in a theory with large distance conformality with the following explicit dependence on the mass and anomalous dimension $\gamma = - \partial \ln m / \partial \ln \Lambda$:  
\begin{eqnarray}
\langle {\Phi_{Q\widetilde{Q} }}^{\bar{N}_f}_{N_f}  \rangle  &\propto& -m \mu^2 \ ,  \qquad ~~~~~~~~~~~~~~~~~~~~~0 <\gamma < 1 \ , \label{BZm} \\
\langle {\Phi_{Q\widetilde{Q} }}^{\bar{N}_f}_{N_f} \rangle   &\propto &   -m \mu^2  \log \frac{\mu^2}{{|\langle {\Phi_{Q\widetilde{Q} }}  \rangle|}^{2/3} }\ , ~~~~~~   \gamma \rightarrow  1    \ , \label{SDm} \\
\langle {\Phi_{Q\widetilde{Q} }}^{\bar{N}_f}_{N_f} \rangle   &\propto &  -m^{\frac{3-\gamma} {1+\gamma}} 
 \mu^{\frac{4\gamma} {1+\gamma}}\ , ~~~~~~~~~~~~~~~~~~~~~~1<\gamma  \leq 2 \ .
   \label{UBm}
 \end{eqnarray}
An interesting reanalysis of the ultraviolet versus infrared dominated components of the vevs above which, however, does not modify these results can be found in \cite{DelDebbio:2010ze}. This shows that for any physically acceptable value of the anomalous dimension the relevant scalar degrees of freedom of the magnetic theory also acquire a mass term, decouple and are expected to develop a vev able to induce masses for the remaining states. In the magnetic theory the dual of $\Phi_{Q\widetilde{Q} }$ is denoted by $\Phi_{m}$. The operator \eqref{inducedmass}, involving $\phi$, $\widetilde{\phi}$ and $\Phi_m$, should be at most  a marginal one at the infrared fixed point for the theory to display large distance conformality. Therefore the physical mass dimension of $\Phi_{m}$  is two and  $\Phi_{Q\widetilde{Q} } = \mu \Phi_{m}$. This implies the fundamental result that the anomalous dimension at the fixed point of the electric quark bilinear cannot exceed one!   

Since the anomalous dimension does not exceed one at the interacting fixed point,   $\Phi_m$ cannot be interpreted as an elementary field here. This is completely analogous to Seiberg's case. This is so since  $\Phi_m$ precisely maps into the auxiliary field of the mesonic chiral superfield in superQCD and therefore cannot propagate. To be able to integrate the $\Phi_m$ field out at the interacting fixed point one needs to add an associated quadratic term leading to the Lagrangian: 
\begin{equation}
 {\cal L}_{\Phi_{Q\widetilde{Q} }} + \frac{\Phi_{Q\widetilde{Q} }^{\ast}\Phi_{Q\widetilde{Q} }}{\mu^2} =  \left( {y_{\Phi_{Q\widetilde{Q} }}}\, \widetilde{\phi}^* \Phi_{ m}{\phi}^*  + {\rm h.c.} \right) + \Phi_m^{\ast}\Phi_m\ . 
\label{njl}
\end{equation}
  Integrating out the auxiliary field $\Phi_m$ leads to the following quartic Lagrangian for the $\phi$ fields: 
  \begin{equation}
  {\cal L}_{\phi^4} =
   - y_{\Phi_{Q\widetilde{Q}}}^2 \, 
   {\widetilde{\phi}^{*r}}_{c1}  \, 
   {\widetilde{\phi}}^{c2}_r  \, 
  \phi_{c2}^{l}  \,
  {\phi^{* c1}}_l \ ,
  \end{equation}
    with $c_i$ the dual color indices and $r$ and $l$ the right and left flavor indices respectively. 
  
Adding an explicit quark mass in the electric theory corresponds to adding the  following  operator:
\begin{equation}
{\rm Tr} \left[ m \Phi_{Q\widetilde{Q}}\right] + {\rm h.c.} = \mu{\rm Tr} \left[ m \Phi_{m}\right] + {\rm h.c.}
 \end{equation}
 in the  magnetic theory which induces a vauum expectation value for the scalar field $\Phi_{Q\widetilde{Q}} = - m \mu^2$ in perfect agreement with the field theoretical result shown in \eqref{BZm}. 

We note that the second term in \eqref{njl}, in the electric variables, is  the Nambu Jona-Lasinio \cite{Nambu:1961tp,Nambu:1961fr} (NJL) four-fermion operator. In the dual variables it means that one can view the ordinary fermionic condensate, at the fixed point, as a composite state of two elementary magnetic scalars. The duality picture offers the first simple explanation of why the anomalous dimension of the fermion condensate does not exceed one at the boundary of the conformal window. 

Finally we comment on the fact that requiring the electric and the magnetic theory to be both asymptotically free  one deduces the following range of possible values of $N_f$: 
\begin{equation}
\frac{3}{2}N \leq  N_f \leq \frac{9}{2} N \ . 
\end{equation}
It is natural to identify this range of values of $N_f$ with the actual extension of the conformal window 

 \section{Conclusion}

We constructed possible magnetic duals of QCD with one adjoint Weyl fermion by classifying 
 {\it all} the solutions of the 't Hooft anomaly matching conditions of the type shown in the table \ref{QCDAdual}. 
 We assumed the number of flavors to be sufficiently large for the electric and magnetic theory to develop an infrared fixed point.  
 The 't Hooft anomaly conditions  constrained the fermionic spectrum and led to a dual gauge group of the type $SU(dN_f - N)$ with $d$ a positive integer. We have shown that the case of $d=1$ leads to the minimal amount of fermionic matter needed to saturate the anomalies and moreover any other choice of $d$ does not allow to move in the flavor space without simultaneously change the number of colors of the electric theory. These results strongly suggest that $SU(N_f - N)$ constitute the obvious candidate for the dual gauge group. 
 
 Imposing consistent flavor decoupling allowed to determine the spectrum of the scalars and the Yukawa sector of the dual theory.
 
 An important result is that we provided a consistent picture for the existence of the first  nonsupesymmetric dual valid for {\it any} number of colors. 
 
 We have also shown that the anomalous dimension of the electric fermion mass operator can never exceed one at the infrared fixed point for the dual theory to be consistent. This is a remarkable result showing that one can obtain important non-perturbative bounds on the anomalous dimension of vector-like nonsupersymmetric gauge theories using duality arguments.
 
 Our dual theory can be already tested with todays first principle lattice techniques. In fact,
given that our theory resembles super QCD but without the fundamental scalars, in the electric theory, establishing the existence of duality in our model is the first step towards checking Seiberg's duality on the lattice.

\appendix
\section{The  Involution Theorem}
\label{involutionApp}
What happens if we dualize again the magnetic theory? The simplest possibility is that one recovers the electric theory. In fact one can imagine  more general situations but for the time being we will assume this to be the case.  The duality transformation is therefore a mathematical {\it involution}. This condition leads to an interesting and general theorem on the gauge structure of any dual gauge group. 
\begin{Def}
Consider an electric $SU(N)$ gauge theory with $N_f$ flavors possessing a magnetic dual constituted of an $SU(X)$ gauge theory with also $N_f$ flavors. The same number of flavors insures that the global symmetries, encoding the physically relevant information of the theory, match in both theories.

The duality transformation from the electric  to the magnetic theory is an {\it involution} if a second duality transformation acting on the magnetic theory gives back the original electric $SU(N)$ theory.
\end{Def}

Denoting the duality transformation with an arrow, we summarize the involution as follows:
\begin{align*}
SU(N) \longrightarrow SU(X) \longrightarrow SU(N)
\end{align*}

We are now ready to enunciate the following theorem: 

\begin{thm}
In order for the $SU(X)$ gauge group to respect {\it involution} we must have:
\begin{equation}  
X = P(N_f) - N \ , 
\end{equation}
where $P(N_f)$ denotes any integer valued polynomial in $N_f$.
\end{thm}

\begin{proof}
We start by noting that the involution property requires $X$ to be linear in $N$.  This is so since, if the duality transformation requires the gauge group $X$ to depend on a generic power $p$ of $N$, i.e. $X \sim N^p$, then a second duality transformation on the magnetic gauge group leads to a second gauge group depending on $(N^{p})^p$. The involution condition requires  $p$ to be one. 

Furthermore, $X$ can still depend on any continuous and integer valued function of $N_f$ via the polynomial $P(N_f) = \sum_{p=0}^{n} \alpha_p N_f^p$.

Combining the two requirements we have that $X = P(N_f) - \beta N$, for some integer $\beta$. Imposing {\it involution} one last time we require:
\begin{align*}
N  = P(N_f) - \beta X  \rightarrow P(N_f) - \beta \left (P(N_f) - \beta N \right ) = \left (1-\beta \right ) P(N_f) + \beta^2 N.
\end{align*}
Only $\beta = 1$ solves this equation non-trivially (i.e. $\beta=-1$ and $P(N_f) = 0$ is the trivial solution),
thus proofing the theorem.
\end{proof}

For the present work, the theorem guarantees that \emph{any} solution
of the 't Hoof anomaly matching conditions above respects the involution condition since $X = d \ N_f - N$.

Note that the {\it involution theorem} corresponds to the minimum requirement for the duality transformation to be an {\it involution}, even before taking into account the specific spectrum of the dual theory.

\end{document}